\documentclass[12pt]{article}

\usepackage{geometry}
\usepackage{lmodern}
\usepackage[T1]{fontenc}
\usepackage[mathletters]{ucs} % Extended unicode (utf-8) support
\usepackage[utf8]{inputenc}
  \usepackage{fourier}

\usepackage{fancyhdr}
\usepackage{graphicx}
\graphicspath{{images/}}

\usepackage{amsmath, amsthm}
\usepackage{amssymb, amsbsy}
\usepackage[dvipsnames]{xcolor}
\usepackage[pagebackref]{hyperref}
\hypersetup{
colorlinks=true,
linkcolor=Brown,
citecolor=OliveGreen,
urlcolor=RoyalBlue,
}

\usepackage{floatflt, wrapfig}
\usepackage{xcolor}

\usepackage{graphicx}
\graphicspath{ {images/} }
\usepackage{epstopdf}
\usepackage{color, import}
\usepackage{sectsty}
\usepackage{enumitem, outlines}
\usepackage{lineno}

\usepackage{cite}
\usepackage[edges]{forest}

\usepackage{amsfonts}
\usepackage{algorithm}

\usepackage{makecell}

\usepackage{setspace}
\usepackage[list=true]{subcaption}
\usepackage[title]{appendix}
\usepackage{multirow}
\usepackage{array}

\usepackage{etoolbox}
\makeatletter
\patchcmd{\@maketitle}{\begin{center}}{\begin{flushleft}}{}{}
\patchcmd{\@maketitle}{\begin{tabular}[t]{c}}{\begin{tabular}[t]{@{}l}}{}{}
\patchcmd{\@maketitle}{\end{center}}{\end{flushleft}}{}{}
\makeatother

\renewenvironment{abstract}
{\small\section*
{\bfseries\noindent{\raisebox{-.15\baselineskip}{\normalsize\abstractname}}\hrulefill} 
}

\newtheorem{theorem}{Theorem}
\newtheorem{lemma}{Lemma}

\theoremstyle{definition}

\author{hkhj\\hjk}
\author{hjjkhj\\hjk}

\begin{document}
\title{Exploring a Dynamic Ring without Landmark}
    \author{
    \large{Archak Das}\\
    \small {Department of Mathematics, Jadavpur University, India}\\
    \small \emph{archakdas.math.rs@jadavpuruniversity.in}\\
    Kaustav Bose\\
    \small {Department of Mathematics, Jadavpur University, India}\\
    \small \emph{kaustavbose27@gmail.com}\\
    Buddhadeb Sau\\
    \small {Department of Mathematics, Jadavpur University, India}\\
    \small \emph{buddhadeb.sau@jadavpuruniversity.in}
    \\
    }
    \date{}
                  % typeset the header of the contribution
    
  \maketitle

\begin{abstract}
Consider a group of autonomous mobile computational entities, called agents, arbitrarily placed at some nodes of a dynamic but always connected ring. The agents neither have any knowledge about the size of the ring nor have a common notion of orientation. We consider the \textsc{Exploration} problem where the agents  have to collaboratively to explore the graph and terminate, with the requirement that each node has to be visited by at least one agent. It has been shown by Di Luna et al. [Distrib. Comput. 2020] that the problem is solvable by two anonymous agents if there is a single observably different node in the ring called landmark node. The problem is unsolvable by any number of anonymous agents in absence of a landmark node. We consider the problem with non-anonymous agents (agents with distinct identifiers) in a ring with no landmark node. The assumption of agents with distinct identifiers is strictly weaker than having a landmark node as the problem is unsolvable by two agents with distinct identifiers in absence of a landmark node. This setting has been recently studied by Mandal et al. [ALGOSENSORS 2020]. There it is shown that the problem is solvable in this setting by three agents assuming that they have edge crossing  detection  capability. Edge crossing  detection  capability is a strong assumption which enables two  agents  moving  in  opposite  directions  through  an  edge in  the  same  round  to  detect  each  other and also exchange information. In this paper we give an algorithm that solves the problem with three agents without the edge crossing  detection  capability. 

% As the basic ingredients of our approach, we solve the subproblems \textsc{Meeting} and \textsc{Contiguous Agreement} which may be of independent interest and useful for solving other problems is similar settings.   

\end{abstract}
\section{Introduction}

Consider a team of autonomous computational entities, usually called agents or robots, located at the nodes of a graph. The agents are able to move from a node to any neighboring node. The \textsc{Exploration} problem asks for a distributed algorithm that allows the agents to explore the graph, with the requirement that each node has to be visited by at least one agent. Being one of the fundamental problems in the field of autonomous multi-agent systems, the problem has been extensively studied in the literature. However, the majority of existing literature studies the problem for \textit{static} graphs, i.e., the topology of the graph does not  change over time.  Recently within the distributed computing community, there has been a surge of interest in highly dynamic graphs: the topology of the graph changes continuously and unpredictably. In highly dynamic graphs, the topological changes are not seen as occasional anomalies (e.g., link failures, congestion, etc) but rather integral part of the nature of the system \cite{DBLP:conf/opodis/Santoro15,DBLP:journals/sigact/KuhnO11}. We refer the readers to \cite{CasteigtsFQS12} for a compendium of different models of dynamic networks considered in the literature.  If time is discrete, i.e., changes occur in rounds, then the evolution of a dynamic graph can be seen as a sequence of static graphs. A popular assumption in this context is \textit{always connected} (Class 9 of \cite{CasteigtsFQS12}), i.e., the graph is connected in each round.

In the dynamic setting, the  \textsc{Exploration} problem was first studied in \cite{dc/LunaDFS20}. In particular, the authors studied the \textsc{Exploration} problem in a dynamic but always connected ring by a set of autonomous agents. They showed that \textsc{Exploration} is solvable by two anonymous agents (agents do not have unique identifiers) under fully synchronous setting (i.e., all agents are active in each round) if there is a single observably different node in the ring called \textit{landmark} node. They also proved that in absence of a landmark node, two agents cannot solve \textsc{Exploration} even if the agents are non-anonymous and they have chirality, i.e., they agree on clockwise and counterclockwise orientation of the ring. The impossibility result holds even if we relax the problem to \textsc{Exploration} with partial termination. As opposed to the standard explicit termination setting where all agents are required to terminate, in the partial termination setting at least one agent is required to detect exploration and terminate.  If the agents are anonymous, then \textsc{Exploration} with partial termination with chirality remains unsolvable in absence of a landmark node even with arbitrary number of agents.  Then in \cite{MandalMM20}, the authors considered the \textsc{Exploration} problem (without chirality and requiring explicit termination) with no landmark node. Since the problem cannot be solved even with arbitrary number of anonymous agents, they considered non-anonymous agents, in particular, agents with unique identifiers. Since the problem is unsolvable by two non-anonymous agents, they considered the question that whether the problem can be solved by three non-anonymous agents. They showed that the answer is yes if the agents are endowed with edge crossing detection capability. Edge crossing detection capability is a strong assumption which enables two agents moving in opposite directions through an edge in the same round to detect each other and also exchange information. In collaborative tasks like exploration, the agents are often required to meet at a node and exchange information. However, the edge crossing detection capability allows two agents to exchange information even without meeting at a node. The assumption is particularly helpful when the agents do not have chirality where it is more difficult to ensure meeting. Even if we do not allow exchange of information, simple detection of the swap can be useful in deducing important information about the progress of an algorithm. In \cite{MandalMM20}, it was  also shown that the assumption of  edge  crossing detection  can be removed with the help of randomness. In particular, without assuming edge crossing detection capability, they gave a randomized algorithm that solves \textsc{Exploration} with explicit termination with probability at least $1 - \frac{1}{n}$ where $n$ is the size of the ring. Therefore this leaves the open question that whether the problem can be solved by a deterministic algorithm by three non-anonymous agents without edge crossing detection capability. In this paper, we answer this question affirmatively.

% \cite{icdcs/LunaDFS16, MandalMM20}

\subsection{Related Work}
The problem of \textsc{Exploration} by mobile agents in static anonymous graph has been studied extensively in the literature \cite{AlbersH00, DasFKNS07, DengP99, DieudonneP14, FlocchiniKMS09, FraigniaudIPPP05, PanaiteP99}. Prior to \cite{dc/LunaDFS20}, there have been a few works on \textsc{Exploration} of dynamic graphs, but under assumptions such as complete a priori knowledge of location and timing of topological changes (i.e., \textit{offline} setting) \cite{MichailS16,IlcinkasW13,Erlebach0K15,IlcinkasKW14} or periodic edges (edges appear periodically) \cite{FlocchiniMS13, IlcinkasW11} or $\delta$-recurrent edges (each edge appears at least once every $\delta$ rounds) \cite{IlcinkasW13} etc. In the \textit{online} or \textit{live} setting where the location and timing of the changes are unknown, distributed \textsc{Exploration} of graphs without any assumption other than being always connected was first considered in \cite{dc/LunaDFS20}. In particular, they considered the problem on an always connected dynamic ring. They proved that  without any knowledge of the size of the ring and without landmark node, \textsc{Exploration} with partial termination is impossible by two agents even if the agents are non-anonymous and have chirality. They also proved that if the agents are anonymous, have no knowledge of size, and there is no landmark node then   \textsc{Exploration} with partial termination is impossible by any number of agents even in the presence of chirality. On the positive side the authors showed that under fully synchronous setting, if an upper bound $N$ on the size of the ring is known to two anonymous agents, then  \textsc{Exploration} with explicit termination is possible within $3N-6$ rounds. They then showed that for two anonymous agents, if chirality and a landmark node is present, then exploration with explicit termination is possible within $O(n)$ round, and in the absence of chirality with all other conditions remaining the same,  \textsc{Exploration} with explicit termination is possible within $O(n \log n)$ rounds, where $n$ is the size of the ring. They have also proved a number of results in the semi-synchronous setting (i.e., not all agents may be active in each round) under different assumptions. Then in \cite{MandalMM20}, the authors considered agents with unique identifiers and edge crossing detection capability in a ring without any landmark node. They showed that \textsc{Exploration} with explicit termination is impossible in the absence of landmark node and the knowledge of $n$ by two agents with access to randomness, even in the presence of chirality, unique identifiers and edge-crossing detection capability. In the absence of randomness even \textsc{Exploration} with partial termination is impossible in the same setting.  With three agents under fully synchronous setting, the authors showed that \textsc{Exploration} with explicit termination is possible  by three non-anonymous agents with edge-crossing detection capability in absence of any landmark node. Removing the assumption of edge-crossing detection and replacing it with access to randomness, the authors gave a randomized algorithm for \textsc{Exploration} with explicit termination with success probability at least $1 - \frac{1}{n}$.  \textsc{Exploration} of an always connected dynamic torus was considered in \cite{DBLP:conf/icdcs/GotohSOKM18}. 
In \cite{GOTOH20211} the problem of \textsc{Perpetual Exploration} (i.e., every node is to be visited infinitely often ) was studied in temporally connected (i.e., may not be always connected but connected over time) graphs. Other problems studied in dynamic graphs include \textsc{Gathering} \cite{LunaFPPSV20, BournatDP18, OoshitaD18}, \textsc{Dispersion} \cite{KshemkalyaniMS20, AgarwallaAMKS18}, \textsc{Patrolling} \cite{DBLP:conf/sofsem/DasLG19} etc.

\subsection{Our Results}

We consider a dynamic but always connected ring of size $n$. A team of three agents are operating in the ring under a fully synchronous scheduler. Each agent has a $k$ bit unique identifier. The agents do not have any knowledge of $n$ and they do not have chirality. Furthermore, they do not have edge crossing detection capability. In this setting, we give a deterministic algorithm for \textsc{Exploration} with explicit termination. The algorithm solves the problem in $O(k^2 2^k n)$ rounds. As a subroutine, we also solve the problem \textsc{Meeting} where any two agents in the team are required to meet each other at a node. Another basic ingredient of our approach is an algorithm for the \textsc{Contiguous Agreement} problem which requires that the agents have to agree on some common direction for some number of consecutive rounds. These problems may be of independent interest and useful for solving other problems is similar settings. A comparison of the results obtained in this paper with previous works is given in Table  \ref{tab}.

\begin{table}[h]
\small
\centering
  \begin{tabular}{| m{1cm} | m{2cm} | m{2.2cm} | m{1.8cm} | m{1.8cm} | m{2.4cm} |}
    \hline
    Paper & \makecell{Number \\of agents} & \makecell{Agents} &  \makecell{Landmark\\node} & \makecell{Edge cross.\\ detection} & \makecell{Algorithm}\\ \hline\hline
    \makecell{\cite{dc/LunaDFS20}} & \makecell{2 } & \makecell{Anonymous}  & \makecell{Yes} & \makecell{No} & \makecell{Deterministic} \\ \hline
    \makecell{\cite{MandalMM20}} & \makecell{3 } & \makecell{Have unique\\ identifiers} &  \makecell{No} & \makecell{Yes} & \makecell{Deterministic} \\ \hline
    \makecell{\cite{MandalMM20}} & \makecell{3 } & \makecell{Have unique\\ identifiers} & \makecell{No} & \makecell{No} & \makecell{Randomized} \\ \hline
    \makecell{This \\ paper} & \makecell{3 } & \makecell{Have unique\\ identifiers}  & \makecell{No} & \makecell{No} & \makecell{Deterministic} \\ \hline
  \end{tabular}
    \caption{Comparison of our results with previous works.}\label{tab}
\end{table}

\subsection{Outline of the Paper}
In Section \ref{secmodel}, we describe the model and terminology used in the paper. In Section \ref{Exprwc}, we give an algorithm for \textsc{Exploration} in the simpler setting where the agents have chirality. In Section \ref{Secwoch}, we use the techniques used in Section \ref{Exprwc} to give an algorithm for \textsc{Exploration} in the absence of chirality.

\section{Model and Terminology}\label{secmodel}

We consider a dynamic ring of size $n$. All nodes of the ring are identical. Each node is connected to its two neighbors via distinctly labeled ports. The labeling of the ports may not be globally
consistent and thus might not provide an orientation. We consider a discrete temporal model i.e., time progresses in rounds. In each round at most one edge of the ring may be missing. Thus the ring is connected in each round. Such a network is known in the literature as a 1-interval connected ring. 

We consider a team of three agents operating in the ring. The agents do not have any knowledge of the size of the ring. Each agent is provided with memory and computational capabilities. An agent can move from one node to a neighbouring node if the edge between them is not missing. Two agents moving in opposite direction on the same edge are not able to detect each other. An agent can only detect an active agent co-located at the same node i.e., if an agent terminates it becomes undetectable by any other agent. Two agents can communicate with each other only when they are present at the same node. Each agent has a unique identifier which is a bit string of length $k > 1$. The length $k$ of the identifier is the same for each agent. For an agent $r$, its unique identifier will be denoted by $r.ID$. Also $val(r.ID)$ will denote the numerical value of $r.ID$. For example $val(00110) = 6$, $val(10011) = 19$, etc. Hence for any agent $r$, $val(r.ID) < 2^{k}$. 

Each agent has a consistent private orientation of the ring, i.e., a consistent notion of left or right. If the left and right of all three agents are the same then we say that the agents have chirality. By clockwise and counterclockwise we shall refer to the orientations of the ring in the usual sense. These terms will be used only for the purpose of description and the agents are unaware of any such global orientation if they do not have chirality. For two agents $r_1$ and $r_2$ on the ring, $d^{\circlearrowright}(r_1, r_2)$ and $d^{\circlearrowleft}(r_1, r_2)$ denotes respectively the clockwise and counterclockwise distance  from $r_1$ to  $r_2$.

% Let us consider a ring of size $n$. Each node is connected to its two neighbors via distinctly labelled ports. The ring is $anonymous$, i.e., the nodes of the ring have no distinctly labelled identifiers. At each timestep one of the edges of the ring may be missing, the choice of which is made by an adversary. Thus the ring is a dynamic ring with 1-interval connectivity dynamism.

% Operating on the ring is a set of three agents. Each agent is equipped with finite memory and computational capabilities. Each agent has an unique ID. The maximum possible value of the ID, say $\Delta$, is known to each agent. Now if $2^{k-1} \leq \Delta < 2^k$, then the ID of each agent can be expressed as $k$-bit strings. None of the agents know anything about the size of the ring. Multiple agents can reside at a single node at the same timestep. Two agents can communicate between themselves only if they are present at the same node.The agents do not have common clockwise or anticlockwise orientation, or $chirality$. The agents can move from one node to another node in the ring at a given timestep if and only if the edge connecting the two nodes is available at that given timestep. Two agents moving in the opposite direction on the same edge in the same timestep is not  able to detect each other. 

We consider a fully synchronous system, i.e., all three agents are active in each round. In each round,  the agents perform the following sequence of operations:

\begin{description}

    \item [\textsc{Look}:]  If other agents are present at the node, then the agent exchanges messages with them.
    
    \item [\textsc{Compute}:] Based on its local observation, memory and received messages, the agent performs some local computations and determines whether to move or not, and if yes, then in which direction. 
    
    \item [\textsc{Move}:] If the agent has decided to move in the \textsc{Compute} phase, then the agent attempts to move in the corresponding direction. It will be able to move only if the corresponding edge is not missing. An agent can detect if it has failed to move.
\end{description}

During the execution of algorithm, two agents can meet each other in two possible ways: (1) two agents $r_1$ and $r_2$ moving in opposite direction come to the same node, or, (2) an agent $r_1$ comes to a node where there is a stationary agent $r_2$. In the second case we say that $r_1$ \textit{catches} $r_2$. If two agents $r_1$, $r_2$ are moving in opposite direction on the same edge in the same round, then we say that $r_1$ and $r_2$ \textit{swaps over an edge}.

\section{Exploration by Agents with Chirality}\label{Exprwc}

In this section, we shall assume that the agents have chirality. Since the agents have agreement in direction we shall use the terms clockwise and counterclockwise instead of right and left respectively. In Section \ref{MeetChiral}  we present an algorithm for \textsc{Meeting} where at least two agents are required to meet at a node. Then in Section \ref{ExplireChiral} we shall use this algorithm as a subroutine to solve \textsc{Exploration}.

\subsection{Meeting by Agents with Chirality}\label{MeetChiral}

We have three agents placed arbitrarily at distinct nodes of the ring. Our objective is that at least two of the agents should meet. The algorithm works in several phases. The lengths of the phases are $2^{j+k}$, $j=0,1,2, \ldots$.  In phase $j$, an agent $r$ tries to move clockwise for the first $val(r.ID)2^{j}$ rounds, and then remains stationary for $(2^{k}-val(r.ID)).2^{j}$ rounds. 

% \begin{theorem}
% There is a round $T_1$ within $\Delta.\sum_{i=1}^{p} 2^i$ rounds from the start of the algorithm, when two agents with state \texttt{search} meet. (Event 1).
% \end{theorem}

We shall prove the correctness of the algorithm in Theorem  \ref{thm: meet chiral}. Before that we shall prove a lemma. This lemma will be used several times in the proofs throughout the paper.

\begin{lemma}\label{lemma main}
Let $r_1, r_2$ and $r_3$ be three agents in the ring such that at round $t$, $0 \leq d^{\circlearrowright}(r_1, r_3) < d^{\circlearrowright}(r_1, r_2)$. If $r_1$  remains static and both $r_2$ and $r_3$ try to move clockwise for the next $2n$ rounds,  then within these $2n$ rounds either $r_2$ meets $r_1$ or $r_3$ meets $r_2$.     
\end{lemma}

\begin{proof}

Assume that at round $t$, $d^{\circlearrowleft}(r_1, r_2) = x$ and $d^{\circlearrowleft}(r_2, r_3) = y$. We have $x + y \leq n$. Within the next $2n$ rounds, if $r_2$ is able to move clockwise for at least $x$ rounds, then it will meet $r_1$ and we are done. So assume that $r_2$ does not meet $r_1$. Suppose that $r_2$ succeeds to make a move $x' < x$ times in the next $2n$ rounds. This means that $r_2$ remains static for $2n - x'$ rounds. Therefore, $r_3$ moves clockwise in those $2n - x'$ rounds when $r_2$ is static. Hence, $d^{\circlearrowleft}(r_2, r_3)$ decreases in these rounds. Recall that initially we had $d^{\circlearrowleft}(r_2, r_3) = y$.   Also, in the $x'$ rounds when $r_2$ was able to move, $d^{\circlearrowleft}(r_2, r_3)$ may or may not have increased depending on whether $r_3$ respectively failed or succeeded to move in those rounds. Now notice that 
\begin{align*}
	2n-x' \geq{}& n + x + y - x'\hspace{1cm}&\text{(since $n \geq x+y$)}\\
		  >{}& n + y  &\text{(since $x > x'$)}\\
		  >{}& x' + y &\text{(since $n > x'$)}
\end{align*}

This implies that $r_3$ will catch $r_2$. This completes the proof that within the $2n$ rounds either $r_2$ meets $r_1$ or $r_3$ meets $r_2$.

\end{proof}

\begin{theorem}\label{thm: meet chiral}
The above algorithm solves \textsc{Meeting} for  three agents with chirality. The algorithm ensures that the meeting takes place within $2^{k+\lceil\log n\rceil + 2}$ rounds.
\end{theorem}

\begin{proof}
Let $p = \lceil\log 2n\rceil = \lceil \log n \rceil + 1$. Hence $p$ is the smallest positive integer such that $2^{p} \geq 2n$. If two agents have met before the $p$th phase, then we are done. If not, we show that a pair of agents is guaranteed to meet during the $p$th phase. Recall that in this phase, an agent $r$  should (attempt to) rotate clockwise for $val(r.ID)2^{p}$ rounds, and then remain stationary for the remaining $(2^k-val(r.ID))2^{p}$ rounds. Let $r_1$ be the first agent to come to rest in that phase, say at round $t$. Let $r_2$ be the agent closest to $r_1$ in the counterclockwise direction and $r_3$ be the third agent.  We have $ {val(r_2.ID)}2^{p} - {val(r_1.ID)}2^{p}  \geq 2n$ and $ {val(r_3.ID)}2^{p} - {val(r_1.ID)}2^{p}  \geq 2n$. This implies that both $r_2$ and $r_3$ attempts to move for at least $2n$ rounds after $r_1$ comes to rest at round $t$. By  Lemma \ref{lemma main}, either $r_2$ and $r_1$ meets, or $r_3$ and $r_2$ meets. So the meeting takes place within $2^k\sum_{i=0}^{p} 2^i < 2^{k+\lceil\log n\rceil + 2}$ rounds.
\end{proof}

\subsection{Exploration with Termination by Agents with Chirality}\label{ExplireChiral}

\begin{algorithm}
  \caption{The algorithm executed by an agent $r$ with state \texttt{search}}
  \begin{outline}
% \renewcommand{\outlinei}{enumerate}

%  \1 Initialize $r.counter \leftarrow 0$
\1[\textbf{1.}] Until another agent is found,  execute the  algorithm for \textsc{Meeting} described in Section \ref{MeetChiral}.

 \1[\textbf{2.}] Upon the first meeting do the following. 
 
  \2[\textbf{2.1}] If two agents are encountered in the first meeting, then set $r.state \leftarrow$ \texttt{bounce} and move counterclockwise. Also initiate the counter $r.TTime$. Otherwise if there is exactly one agent $r'$ then do the following.

  \2[\textbf{2.2.}] If $r'.state = $ \texttt{search} and $r'.winner = False$ then do the following. If $val(r.ID) < val(r'.ID)$, then set $r.state \leftarrow$ \texttt{settled} and remain at the current node. Otherwise if $val(r.ID) > val(r'.ID)$, then set $r.winner \leftarrow True$ and keep moving clockwise until the next meeting. Also initiate a counter $SCount$ to count the number of successful steps since the first meeting with $r'$.
  
  \2[\textbf{2.3.}] If $r'.state = $ \texttt{search} and $r'.winner = True$,  then set $r.state \leftarrow$ \texttt{bounce} and move counterclockwise. Also initiate the counter $r.TTime$.
  
%   \2 If $r'.state = $ \texttt{search}, $r'.winner = True$ and $r'.RSize \neq \emptyset$ then terminate.
  
 \2[\textbf{2.4.}] If $r'.state = $ \texttt{settled}, then keep moving clockwise until the next meeting.

\1[\textbf{3.}] Upon any subsequent meeting do the following.

  \2[\textbf{3.1.}] If only one agent $r'$ is encountered with  $r'.state = $ \texttt{search}, then do the following.
  
    \3[\textbf{3.1.1.}] If $r'.winner = False$ and $r.RSize = \emptyset$, then set $r.state \leftarrow$ \texttt{forward} and move clockwise. Also initiate the counter $r.TTime$.
    
    \3[\textbf{3.1.2.}] If $r'.winner = True$ and $r'.RSize = \emptyset$, then set $r.state \leftarrow$ \texttt{bounce} and move counterclockwise. Also initiate the counter $r.TTime$. Otherwise, if $r'.winner = True$ and $r'.RSize \neq \emptyset$, then terminate.

   \2[\textbf{3.2.}] If only one agent $r'$ is encountered with $r'.state = $ \texttt{settled}, then do the following.
  
    \3[\textbf{3.2.1.}] If $r.winner = True$, then set $r.RSize \leftarrow  r.SCount = n$ (since $r'$ is encountered for the second time) and inform $r'$ about $n$. Keep moving clockwise for $2n$ more rounds and then terminate.
    
    \3[\textbf{3.2.2.}] If $r.winner = False$ and $r'.RSize \neq \emptyset$, then terminate.
  
   \2[\textbf{3.3}] If two agents are encountered then there must be an agent $r'$ with $state$ \texttt{search}. Then execute \textbf{3.1.1.} or \textbf{3.1.2.} whichever is applicable.

%   \2 If $r'.state = $ \texttt{search},  $r'.winner = True$ and $r.winner = False$ then do the following. 
  
%     \3 If none of $r$ and $r'$ know $size$, then  set $r.state \leftarrow$ \texttt{bounce} and move counterclockwise.
    
%     \3 If at least one of $r$ and $r'$ knows $size = n$, then communicate so that both know $n$ $size = n$. Then set $r.acknowledge \leftarrow True$ and try to move counterclockwise for next $n$ rounds and then terminate if no other meeting occurs in between. 

%   \2 If $r'.state = $ \texttt{settled} and $r.acknowledge = True$, ask $r'$ to terminate and then terminate.
    
%   \2 If $r'.state = $ \texttt{settled}, $r.acknowledge = False$ and $r.winner = True$, then do the following.
  
%     \3 If $r'$ is encountered for the second time, infer $size \leftarrow n$ and keep moving clockwise.
    
%     \3 If $r'$ is encountered for the third time, ask $r'$ to terminate and then terminate.
    
%   \2 If $r'.state = $ \texttt{settled}, $r.acknowledge = False$ and $r.winner = False$, then do the following.
  
%     \3 If $r'$ is encountered for the first time, and keep moving clockwise.
    
%     \3 If $r'$ is encountered for the second time, infer $size \leftarrow n$, try to move clockwise for next $n$ rounds and then terminate if no other meeting occurs in between. 
\end{outline}\label{Algor 1}
\end{algorithm}

\begin{algorithm}
  \caption{The algorithm executed by an agent $r$ with state \texttt{settled}}
  \begin{outline}
\renewcommand{\outlinei}{enumerate}

\1[\textbf{1.}] Do not move. Terminate on one of the following conditions.

    \2[\textbf{1.1.}] If the other two agents are present at the same time and one of them has $state$ \texttt{forward}, then terminate immediately.  

    \2[\textbf{1.2.}]  If an agent $r'$ is encountered such that $r'.SBound \neq \emptyset$, then terminate immediately.
    
    \2[\textbf{1.3.}] If $n$ is already known, i.e.,  $r.RSize \neq \emptyset$, and an agent $r'$ is encountered that does not know $n$, i.e., $r'.RSize = \emptyset$, then terminate immediately.
    
    \2[\textbf{1.4.}] If an agent $r'$ with state \texttt{search} informs $r'.RSize = n$, then terminate after $2n$ more rounds.
    
    \2[\textbf{1.5.}] If an agent $r'$ with state \texttt{forward} or \texttt{return} informs $r'.RSize = n$ and $r'.TTime$, then initiate counter $r.TTime$ with starting value $r.TTime \leftarrow r'.TTime$ and terminate after the round when $r.TTime = 16n$.

\end{outline}\label{Algor 2}
\end{algorithm}

\begin{algorithm}
  \caption{The algorithm executed by an agent $r$ with state \texttt{forward}}
  \begin{outline}
\renewcommand{\outlinei}{enumerate}

%  \1 Initialize $r.counter \leftarrow 0$
 
 \1[\textbf{1.}] Until a new meeting, keep moving clockwise. Maintain the counters $r.FSteps \leftarrow $ number of successful steps with state \texttt{forward} and $r.TTime \leftarrow $ number of rounds since the triple was formed.

%  Increment $r.counter$ every round (to count the number of rounds elapsed until an agent met). 

%  \1  Upon meeting an agent, say $r'$, set $r.time\_stamp \leftarrow r.counter$ and do the following.
 
%  and change $r.state$ according to the following rule.

\1[\textbf{2.}] If all three agents meet at a node then terminate immediately. Otherwise upon meeting an agent $r'$, do the following. 

    \2[\textbf{2.1.}] If $r.RSize = \emptyset  \wedge r'.RSize \neq \emptyset$ or $r.RSize \neq \emptyset  \wedge r'.RSize = \emptyset$ then terminate immediately. Otherwise, do the following.

    \2[\textbf{2.2.}] If $r'.state =$ \texttt{settled}, then infer $r.RSize \leftarrow r.SCount = n$ (since $r'$ is encountered for the second time) and inform $r'$ about $n$. Keep moving clockwise and terminate after the round when $r.TTime = 16n$.

    \2[\textbf{2.3.}] If $r'.state =$ \texttt{bounce}, set $r.SBound \leftarrow r.FSteps + r'.BSteps$. Then keep moving clockwise for $r.SBound$ rounds and then terminate. If the \texttt{settled} agent is met meanwhile, then terminate immediately.

    \2[\textbf{2.4.}] If $r'.state =$ \texttt{return}, then do the following.
    
        \3[\textbf{2.4.1.}] If $r$ catches $r'$, then set $r.SBound \leftarrow r.FSteps + r'.BSteps$. Then keep moving clockwise for $r.SBound$ rounds and then terminate. If the \texttt{settled} agent is met meanwhile, then terminate immediately.
        
        \3[\textbf{2.4.2.}] If $r'$ catches $r$ and $r'.Rsteps > 2r'.Bsteps$, then keep moving clockwise until the next meeting. Otherwise if $r'.Rsteps \leq 2r'.Bsteps$, then set $r.SBound \leftarrow r.FSteps + r'.BSteps + 1$. Then keep moving clockwise for $r.SBound$ rounds and then terminate. If the \texttt{settled} agent is met meanwhile, then terminate immediately.

\end{outline}\label{Algor 3}
\end{algorithm}

\begin{algorithm}
  \caption{The algorithm executed by an agent $r$ with state \texttt{bounce}}
  \begin{outline}
\renewcommand{\outlinei}{enumerate}

%  \1 Initialize $r.counter \leftarrow 0$
 
 \1[\textbf{1.}] Maintain the counters $r.BSteps \leftarrow $ number of successful steps with state \texttt{bounce} in the current run, $r.BBlocked \leftarrow $ number of unsuccessful attempts with state \texttt{bounce} in the current run and $r.TTime \leftarrow $ number of rounds since the triple was formed. Until a new meeting or fulfillment of the condition $r.BBlocked > r.BSteps$, keep moving in the counterclockwise direction. If $r.BBlocked > r.BSteps$ is satisfied, change $r.state \leftarrow \texttt{return}$ and move clockwise.

\1[\textbf{2.}] If all three agents meet at a node then terminate immediately. Otherwise upon meeting any agent $r'$, do the following.

    \2[\textbf{2.1.}] If $r.RSize = \emptyset  \wedge r'.RSize \neq \emptyset$ or $r.RSize \neq \emptyset  \wedge r'.RSize = \emptyset$ then terminate immediately. Otherwise, do the following.

    \2[\textbf{2.2.}] If $r'.state =$ \texttt{settled} then do the following.
    
        \3[\textbf{2.2.1.}] If it is the first meeting with $r'$ in the same run, then keep moving counterclockwise until a new meeting or fulfillment of the condition $r.BBlocked > r.BSteps$. Also initiate a counter $SCount$ to count the number of successful steps since the first meeting with $r'$ in the current run.
        
        \3[\textbf{2.2.2.}] If it is the second meeting with $r'$ in the same run, then set $r.RSize \leftarrow r.SCount = n$. Change $r.state \leftarrow \texttt{return}$ and move clockwise.
        
        % keep moving clockwise and terminate after the round when $r.TTime = 13n$. 

    \2[\textbf{2.3.}] If $r'.state =$ \texttt{forward}, then set $r.SBound \leftarrow r'.FSteps + r.BSteps$. Then keep moving counterclockwise for $r.SBound$ rounds and then terminate. If the \texttt{settled} agent is met meanwhile, then terminate immediately.

\end{outline}\label{Algor 4}
\end{algorithm}

\begin{algorithm}
  \caption{The algorithm executed by an agent $r$ with state \texttt{\texttt{return}}}
  \begin{outline}
\renewcommand{\outlinei}{enumerate}

%  \1 Initialize $r.counter \leftarrow 0$
 
 \1[\textbf{1.}]  If the size of the ring is already known, i.e.,  $r.RSize = n$, then keep moving clockwise and terminate after the round when $r.TTime = 16n$. Otherwise, keep moving in the clockwise direction until a new meeting. Maintain the counters $r.RSteps \leftarrow $ number of successful steps with state \texttt{return} in the current run and $r.TTime \leftarrow $ number of rounds since the triple was formed.

\1[\textbf{2.}] If all three agents meet at a node then terminate immediately. Otherwise upon meeting any agent $r'$, do the following.

    \2[\textbf{2.1.}] If $r.RSize = \emptyset  \wedge r'.RSize \neq \emptyset$ or $r.RSize \neq \emptyset  \wedge r'.RSize = \emptyset$ then terminate immediately. Otherwise, do the following.

    % \2[\textbf{2.2.}] If $r'.state =$ \texttt{settled} then do the following.
    
    %     \3[\textbf{2.2.1.}] If $r'.RSize = \emptyset$, then keep moving counterclockwise until a new meeting.
        
    %     \3[\textbf{2.2.2.}] If $r'.RSize \neq \emptyset$, then set $r.RSize \leftarrow n$. Keep moving clockwise and terminate after the round when $r.TTime = 16n$. 

    \2[\textbf{2.2.}] If $r'.state =$ \texttt{forward}, then do the following. 
        
        \3[\textbf{2.2.1.}] If $r'$ catches $r$, then set  $r.SBound \leftarrow r.FSteps + r'.BSteps$. Then move counterclockwise for $r.SBound$ rounds and then terminate. If the \texttt{settled} agent is met meanwhile, then terminate immediately.
        
         \3[\textbf{2.2.2.}] If $r$ catches $r'$ then do the following.

            \4[\textbf{2.2.2.1.}] If $r.RSteps \leq 2r.BSteps$ then set  $r.SBound \leftarrow r.FSteps + r'.BSteps +1$. Then move counterclockwise for $r.SBound$ rounds and then terminate. If the \texttt{settled} agent is met meanwhile, then terminate immediately.
         
            \4[\textbf{2.2.2.2.}] Otherwise, if $r.RSteps > 2r.BSteps$, then change $r.state \leftarrow \texttt{bounce}$ and move counterclockwise.

        % \3 If $r'$ catches $r$ and $r'.Rsteps > 2r_3.Bsteps$, then keep moving clockwise until the next meeting. Otherwise if $r'.Rsteps \leq 2r'.Bsteps$, then set $r.SBound \leftarrow r.FSteps + r'.BSteps + 1$. Then keep moving clockwise for $r.SBound$ rounds and then terminate. If the \texttt{settled} agent is met meanwhile, then ask it to terminate.

% \1 Upon meeting the landmark agent twice note $n$.

% \1 If $n$ is known, then if $Ntime > 3n$, then terminate

% \1 If caught by \texttt{F}, then calculate $k = Fsteps + Bsteps$ and after that change direction to anti-clockwise and start moving until landmark is reached or te  $k$ rounds have elapsed, whichever earlier and then terminate.

% \1 Upon catching \texttt{F}, then $r.role \leftarrow $ \texttt{RC}.

\end{outline}\label{Algor 5}
\end{algorithm}

\subsubsection{Description of the Algorithm}

We consider three agents in the ring having chirality. For simplicity assume that the agents are  initially placed at distinct nodes of the ring. We shall later remove this assumption. Our plan is to first bring two of the agents at the same node using the \textsc{Meeting} algorithm described in Section \ref{MeetChiral}. Then one of them will settle  at that node and play the role of landmark node. Then the situation reduces to a setting  similar to  \cite{dc/LunaDFS20}. However we cannot use the same algorithm from \cite{dc/LunaDFS20}  in our case. This is because unlike in \cite{dc/LunaDFS20} we have to ensure that the agent acting as landmark also terminates. However our algorithm uses some ideas from \cite{dc/LunaDFS20}. We now describe our algorithm in the following. The detailed pseudocode description of the algorithms are given in Algorithm \ref{Algor 1}, \ref{Algor 2}, \ref{Algor 3}, \ref{Algor 4} and \ref{Algor 5}.

Initially all the agents start with their $state$ variable set to \texttt{search}.  Until an agent meets another agent, it executes the \textsc{Meeting} algorithm described in Section \ref{MeetChiral}.  Now according to Theorem \ref{thm: meet chiral},  two agents are guaranteed to meet within $2^{k+\lceil\log 2n\rceil + 2}$ rounds from the start of the algorithm. On meeting the agents compare their IDs and the one with smaller ID changes its $state$ to \texttt{settled} and stops moving. The other agent changes its $winner$ variable to $True$ and henceforth abandons its phase-wise movement and attempts to move  clockwise in each round.

% Initially all the agents have a variable $winner$ set to $false$. Now, when an agent $r$ meets with an agent $r'$ during the course of its movement and it is the first meeting of both the agents, i.e., both $r.winner$ and $r'.winner$ is set to $false$. Let us assume without any loss of generality that $val(r.ID) < val(r'.ID)$. Then $r$ changes its $state$ to \texttt{settled} and henceforward do not move from its position. The agent $r'$ changes  $r'.winner$ to $true$ and from now onwards $r'$ abandons its phase-wise movement and attempts to clockwise in each round.

% In addition to that $r'$ uses the variable $RSize$ to calculate the size of the ring the next time it encounters $r$ (which is now the \texttt{settled} agent).  

Let us now describe the case when an agent with $state$ \texttt{search} meets the \texttt{settled} agent. If an agent  with $winner = False$ encounters the \texttt{settled} agent it also abandons its phase-wise movement and henceforth tries to move in the clockwise direction in every round. If an agent  with $winner = True$ meets the \texttt{settled} agent $r$, then it indicates that it is meeting the \texttt{settled} agent for the second time and hence all  nodes of the ring have been explored. The agent can also calculate the size of the ring as it is equal to the number of successful moves between the two meetings. The agent assigns this value to the variable $RSize$ and also informs the \texttt{settled} agent about it. Then the agent will continue to move in the clockwise direction for $2n$ more rounds. Both these agents will terminate after the completion of these $2n$ rounds. 

Now consider the case when an agent with $state = \texttt{search}$ and $winner = True$ meets an agent with $state = \texttt{search}$ and $winner = False$. If the agent with $winner = True$ already knows $n$, i.e., it has visited the \texttt{settled} agent twice, then both of them terminates immediately. If the agent with $winner = True$ does not already know $n$, then it changes its $state$ to \texttt{forward} and continues to move in the clockwise direction every round. On the other hand, the agent with $winner = False$ changes its $state$ to \texttt{bounce} and starts moving in the counterclockwise direction. This phenomenon is called the formation of \texttt{settled}-\texttt{forward}-\texttt{bounce} triplet. In this case, both the agents initiates a variable $TTime$ to keep track of the number of rounds elapsed after triplet formation.

After the triplet is created, the agent with $state$ \texttt{forward} will continue to move in clockwise direction. The agent with $state$ \texttt{bounce} will move counterclockwise and then on fulfillment of certain conditions, it may change its $state$ to \texttt{return} and start moving clockwise. Then it may again change its $state$ to \texttt{bounce} and start moving counterclockwise. The period between any two such $state$ changes will be called a $run$. While moving in the clockwise direction with $state$ \texttt{forward}, the agent keeps count of the number of successful steps with state \texttt{forward} in the variable $FSteps$. The variable $BSteps$ (resp. $RSteps$) is used to keep count of the number of successful steps with state \texttt{bounce} (resp. \texttt{return}) in the current run. Also while moving in the counterclockwise direction with $state$ \texttt{bounce}, the variable $BBlocked$ counts the number of unsuccessful attempts to move in that run.

% While moving in the counterclockwise direction, the agent with $state$ \texttt{bounce} will keep count of the number of successful steps in the current run with $state$ \texttt{bounce} in the variable $BSteps$ and also the number of unsuccessful attempts to move in the variable $BBlocked$.

An agent $r$ with $state$ \texttt{bounce} will change its $state$ to \texttt{return} if one of the following takes place: 1)  $r.BBlocked$ exceeds $r.BSteps$ or 2) the agent $r$ encounters  the \texttt{settled} agent  twice in the same run. An agent $r$ with $state$ \texttt{return} will change its $state$ to \texttt{bounce} if $r$ meets with the agent with $state$ \texttt{forward} and $r.RSteps > 2r.BSteps$, where $BSteps$ was counted  in the last run with $state$ \texttt{bounce}. 

Here the main idea is that  the agents will try to gauge the size of the ring. An agent may be able to  find the size $n$  exactly or calculate an upper bound of $n$. An agent can exactly find $n$ only if it visits the static \texttt{settled} agent twice in the same direction. In this case it will also inform the \texttt{settled} agent about $n$. Clearly when this happens the ring has been explored completely. However the two agents cannot terminate immediately because the third agent is not aware of this. So the agents will remain active till $TTime = 16n$, i.e., $16n$ rounds from the time when the triplet was created. It should be noted here that the \texttt{settled} agent initially did not know the time when the triplet was created. It came to know about this from the $TTime$ value of some agent that it met and initiated its own $TTime$ counter accordingly. Now it can be shown that within these $16n$ rounds the third agent will meet one of the two agents that already know $n$. These two agents will terminate immediately upon meeting. Now consider the case where an agent is able to find an upper bound of $n$. This happens when one of the following three takes place: 1) the \texttt{forward} agent  meets the agent  with $state$ \texttt{bounce}, 2) the \texttt{forward} agent catches the agent with $state$ \texttt{return}, 3) the agent with $state$ \texttt{return} catches the \texttt{forward} agent with $RSteps \leq 2BSteps$. It can be shown in each of the cases, these two agents will be able to correctly calculate an upper bound $SBound$ of $n$. Furthermore these cases imply that the ring has been already explored completely. However the two agents cannot terminate immediately because the \texttt{settled} agent is not aware of this. Therefore in order to acknowledge the \texttt{settled} agent, these two agents will start moving in opposite directions for $SBound$ more rounds and then terminate. Clearly one of them will be able to meet the \texttt{settled} agent.

\subsubsection{Correctness}

\begin{lemma}\label{lemma T_1}
There exists a round $T_1 \leq$ $2^{k+\lceil\log n\rceil + 2}$ when two of the agents with state \texttt{search} meet and a \texttt{settled} agent is created.
\end{lemma}

\begin{proof}
This follows from Theorem \ref{thm: meet chiral} since the agents execute the algorithm from Section \ref{MeetChiral} until they meet another agent.
\end{proof}

\begin{lemma}\label{lemma T2}
Suppose that $r_1$ and $r_2$ meets at round $T_1$ and $r_1$ becomes \texttt{settled}. There is a round $T_2$, with $T_1 < T_2 <  2^{k+\lceil\log n\rceil + 3}$,  when the third agent $r_3$ meets either $r_1$ or $r_2$.
\end{lemma}

% \begin{lemma}
% Let $T_1$ denote the round when for the first time two agents in the network meet. Then after finitely many rounds say at round $T_2$ all the agents should know about the existence of the landmark agent.
% \end{lemma}

\begin{proof}
At round $T_1$, $r_1$ and $r_2$ meet and $r_1$ becomes \texttt{settled}. After $T_1$, $r_2$ is trying to move clockwise in each round. On the other hand, since $r_3$ is still executing the \textsc{Meeting} algorithm, it will try to move clockwise on some rounds and on other rounds, will not try to move at all. Now if both $r_2$ and $r_3$ try to move clockwise for some $2n$ consecutive rounds  together, then by Lemma \ref{lemma main} either $r_2$ meets $r_3$ or $r_3$ meets $r_1$. Clearly, this is guaranteed to happen in any phase $l$, where $l \geq p = \lceil\log 2n\rceil$. Now suppose that $T_1$ belongs to the $j$th phase. By Lemma \ref{lemma T_1}, $j \leq p$. If $j < p$, then the required meeting should take place in or before the $p$th phase and if $j = p$, then the required meeting will take place in $p$th or $(p+1)$th phase. Therefore, we have $T_2 \leq \sum_{i=0}^{p+1} 2^{i+k} < 2^{k+\lceil\log n\rceil + 3}$.
\end{proof}

% $2^{j+k} + 2^{j+k+1}+......+ 2^{p+k} < $ . 

% Let $l$ be the least such integer, such that $l \geq j$, and $2^l > n$.
% So, in the $l$th phase, both $r_2$ and $r_3$ are active for greater than $n$ rounds.
% Hence either $r_3$ meets $r_1$, or $r_3$ meets $r_2$ in the $l$th phase. In either case $r_3$ comes to know about the presence of landmark agent.

% Now, the maximum time interval between $T_1$ and $T_2$ $\leq$ $2^{j+k} + 2^{j+k+1}+......+ 2^{k+l}$.
% Hence after finitely many rounds after $T_1$, all the agents comes to know about the existence of landmark agent.

\begin{lemma}\label{lemma four}
Within $T_2 + 4n$ rounds either all three agents terminate or a \texttt{settled}-\texttt{forward}-\texttt{bounce} triple is created.
\end{lemma}

\begin{proof}
Recall that by Lemma \ref{lemma T2}, at round $T_2$, either $r_3$ meets  $r_1$ or $r_3$ meets $r_2$. In the latter case, a \texttt{settled}-\texttt{forward}-\texttt{bounce} triple is created and we are done. So consider the other case where $r_3$ meets the \texttt{settled} agent $r_1$. 
Before this meeting, $r_3$ was trying to move clockwise in some rounds, while in other rounds, it was not trying to move at all. After meeting $r_1$, $r_3$ will try moving clockwise in each round. Then by Lemma \ref{lemma main} within $2n$ rounds (i.e., within $T_2 + 2n$ rounds from the beginning), either $r_2$ meets $r_1$ again, or $r_3$ meets $r_2$.  In the latter case, a \texttt{settled}-\texttt{forward}-\texttt{bounce} triple is created, and we are done. So consider the other case where $r_2$ meets $r_1$ for the second time. Thus $r_2$ finds out $n$ (the size of the ring) and also informs $r_1$ about it. Then $r_2$ will keep moving clockwise for $2n$ more rounds and then terminate. Also $r_1$ will remain active for $2n$ more rounds and then terminate. Now within these $2n$ rounds, $r_3$ will meet either $r_1$ or $r_2$ and terminate immediately. Therefore,  within $T_2 + 4n$ rounds either all three agents terminate or a \texttt{settled}-\texttt{forward}-\texttt{bounce} triple is created.
\end{proof}

Suppose that at round $T_3 \leq T_2 + 4n$ a \texttt{settled}-\texttt{forward}-\texttt{bounce} triple is formed. We shall denote by $r_1$, $r_2$ and $r_3$ the agents with states \texttt{settled}, \texttt{forward} and \texttt{bounce} respectively. We shall say that the agents $r_2$ and $r_3$  \textit{agree on an upper bound of} $n$ if one of the following events occur at any round after $T_3$. We show in Lemma \ref{lemma events} that indeed if one of the following events take place then the agents can find an upper bound of $n$.

\begin{description}
\item[Event 1.] The agent $r_2$ with state \texttt{forward} and $r_3$ with state \texttt{bounce} meet each other at a node and neither of them know $n$.

\item[Event 2.] The agent $r_2$ with state \texttt{forward} catches $r_3$ with state \texttt{return} at a node and neither of them know $n$.

\item[Event 3.] The agent $r_3$ with state \texttt{return} catches $r_2$ with state \texttt{forward} at a node and $r_3.Rsteps \leq 2r_3.Bsteps$ and neither of them know $n$.

% \item[Event 4.] The agent $r_3$ with state \texttt{return} catches $r_2$ with state \texttt{forward} at a node and $r_3$ knows $n$
\end{description}

% The agent $r_3$ with state \texttt{bounce} (resp. \texttt{return}) and moving in counterclockwise (resp. clockwise) direction may after some time change its state to \texttt{return} (resp. \texttt{bounce}) and its direction clockwise (resp. counterclockwise). The time between two such state changes is called a \textit{run} of $r_3$.

\begin{lemma}\label{lemma events}
If one of Event 1-3 takes place, then 

\begin{enumerate}
    \item exploration is complete,
    
    \item $r_2$ and $r_3$ can infer an upper bound $N$ of $n$ such that $n \leq N \leq 3n$ and will set it as $SBound$.
\end{enumerate}

\end{lemma}

\begin{proof}
\textbf{Event 1.} Consider a run of $r_3$, starting when it met $r_2$ and changed its state to \texttt{bounce} and ending when it met $r_2$ again (Event 1). Then it is easy to see that all nodes of the ring have been explored by $r_2$ and $r_3$.  Upon meeting, both $r_2$ and $r_3$ set $SBound = r_{2}.FSteps + r_3.BSteps$. Clearly $SBound \geq n$. Also since both $r_2$ and $r_3$ do not know $n$, $r_2.FSteps < n$  and $r_3.BSteps < 2n$. To see the first inequality observe that if $r_2. FSteps \geq n$ then it implies that $r_2$ with $state$ \texttt{forward} has met $r_1$. But recall that $r_2$ has already met $r_1$ before when it was in $state$ \texttt{search} and moving in the same direction. This implies that $r_2$ can infer $n$ by counting the number of successful steps since the first meeting. For the second inequality observe that if $r_3.BSteps \geq 2n$ then it implies $r_3$  has met $r_1$ twice in the same run and therefore can infer $n$ by counting the number of successful steps from the first and second meeting.  Hence $SBound < 3n$.

\textbf{Event 2.} Suppose that at some node $u$, $r_2$ and $r_3$ meet each other and continue moving in clockwise (with state \texttt{forward}) and counterclockwise (with state \texttt{bounce}) direction respectively. Then suppose that at node $v$, $r_3$ changes its state to \texttt{return} and its direction to clockwise. Then after sometime $r_2$ catches $r_3$ (Event 2). Clearly all nodes in the counterclockwise path from $u$ to $v$   have been visited by $r_3$ and all nodes in the clockwise path from $u$ to $v$ have been visited by $r_2$. Hence all nodes in the ring have been together explored by $r_2$ and $r_3$. Upon meeting, they both set $SBound = r_{2}.FSteps + r_3.BSteps$. Clearly $r_3.BSteps = d^{\circlearrowleft}(u, v)$ and $r_2.FStep \geq d^{\circlearrowright}(u, v)$. Hence $SBound \geq d^{\circlearrowright}(u, v) + d^{\circlearrowleft}(u, v) = n $. Also by previous argument $SBound < 3n$.

\textbf{Event 3.} Suppose that at some time $r_2$ and $r_3$ meet each other at node $u$ and continue moving in clockwise (with state \texttt{forward}) and counterclockwise (with state \texttt{bounce}) direction respectively. Then at some time $r_3$ changes its state to \texttt{return} (when we have $r_3.BBlocked = r_3.BSteps + 1$) and its direction to clockwise. Then after sometime it catches $r_2$ and finds that $r_3.RSteps \leq 2r_3.BSteps$ (Event 3). We show that this implies that exploration is complete.
If $r_2$ and $r_3$ swapped over an edge at some  round in between, then it implies that all nodes of the ring have been explored and we are done.  Otherwise we have  $r_3.RSteps = r_3.BSteps + r_2.FSteps$ at the time of meeting. Since $r_3.RSteps \leq 2r_3.BSteps$, we have $r_2.FSteps \leq r_3.BSteps$. Also recall that $r_3.BBlocked = r_3.BSteps + 1$. Now if $r_2$ had been able to successfully execute a move  during each of those rounds when $r_3$ was blocked with state \texttt{bounce}, we must have had $r_2.FSteps > r_3.BSteps$. But since $r_2.FSteps \leq r_3.BSteps$, $r_2$ with state \texttt{forward} and $r_3$ with state \texttt{bounce} must have been blocked at the same round. This can only happen if $r_2$ and $r_3$ are blocked on two ends of the same missing edge. This implies that the ring has been explored. Upon meeting both $r_2$ and $r_3$ set $SBound = r_{2}.FSteps + r_3.BSteps + 1$. If they had swapped over an edge then clearly $SBound \geq n$. Otherwise we showed that  $r_2$ and $r_3$ were blocked at two adjacent nodes of the ring, say  $v$ and $w$ respectively. Then $r_3.BSteps \geq d^{\circlearrowleft}(u, w)$ and $r_2.FStep \geq d^{\circlearrowright}(u, v)$. Hence $SBound \geq d^{\circlearrowright}(u, v) + d^{\circlearrowleft}(u, w) + 1= n$. Also by previous argument $SBound < 3n + 1 \leq 3n$.

\end{proof}

% \begin{lemma}\label{lemma bouncen}
% If $r_3$ with state \texttt{bounce} finds out $n$ at some round $T$ by meeting $r_1$ twice in the same run, then by round $T + 2n$ all three agents find out $n$. 
% \end{lemma}

% \begin{proof}
% At round $T$, both $r_1$ and $r_3$ find out $n$. If $r_2$ already knows $n$, then their is nothing to prove. So assume otherwise. Notice that from round $T$ onwards, both $r_2$ and $r_3$ are moving clockwise since $r_3$ will change its direction at round $T$. Also none of the three agents terminate within next $2n$ rounds. So within next $2n$ rounds either $r_3$ will catch $r_2$ or $r_2$ meets $r_1$. In both cases, $r_2$ comes to know $n$. 
% \end{proof}

% \begin{lemma}\label{lemma bound}
% If at some round $T$, one of Event 1-4 takes place and $r_2, r_3$ sets a value to $SBound$. Then within  
% \end{lemma}

\begin{lemma}\label{lemma 7n}
Suppose that $r_2$ and $r_3$ meet at some round $T$ and $r_3$ changes its $state$ to \texttt{bounce}. If they do not meet again for $10n$ rounds then all three agents are guaranteed to find out $n$
\end{lemma}

\begin{proof}
Assume that $r_2$ and $r_3$ do not meet for $10n$ rounds from $T$. Now at round $T$, $r_3$ changes its state to \texttt{bounce}.  By round $T + 4n$ either $r_3$ has made $2n$ successful moves in the counterclockwise direction or the condition $BBlocked > BSteps$ is fulfilled. In the former case $r_3$ finds out $n$ and changes its state to \texttt{return}. In the latter case also $r_3$ changes its state to \texttt{return}. Therefore within $4n$ rounds $r_3$ is guaranteed to start moving in clockwise direction, say at round $T' \leq T + 4n$ i.e., from $T'$ onwards both $r_2$ and $r_3$ are moving in the same direction. Therefore within $2n$ rounds  either $r_2$ and $r_3$ meet each other or one of them meets $r_1$. Since we assumed that $r_2$ and $r_3$ do not meet, the latter takes place. 

\textbf{Case 1:} Suppose that $r_1$ meets $r_2$. Then both $r_1$ and $r_2$ find out $n$ (if not already known). Within $2n$ rounds $r_3$ comes to $r_1$  and also find out $n$ (if not already known). So within $4n$ rounds from $T'$, i.e. $8n$ rounds from $T$ all three agents find out $n$.

\textbf{Case 2:} Suppose that $r_1$ meets $r_3$. Then within next $2n$ rounds  $r_2$ meets $r_1$. Then by the arguments of Case 1, all three agents will find out $n$ within next $2n$ rounds. Therefore, within  $6n$ rounds from $T'$, i.e., $10n$ rounds from $T$, all three agents find out $n$. 
\end{proof}

\begin{lemma}\label{sixteen}
Within $16n$ rounds after $T_3$, either all three agents find out $n$ or $r_2, r_3$ agree on an upper bound of $n$.
\end{lemma}

\begin{proof}
Let us refer to the meeting of $r_2$ and $r_3$ at round $T_3$ as their $1$st meeting. If $r_2$ and $r_3$ do not meet each other again for $16n$ rounds, then by Lemma \ref{lemma 7n} all three agents will find out $n$, and we are done. Therefore we assume that $r_2$ and $r_3$ meet again at least once after $T_3$ within $16n$ rounds. Let $t_i$ denote the time between the $i$th and $(i+1)$th meeting. In view of Lemma \ref{lemma 7n} we can assume that $t_i \leq 10n$ because otherwise all three agents will find out $n$, and we are done. Also if any one of these meetings is of the type Event 1-3, then we are done as such meetings allow $r_2$ and $r_3$ to agree on an upper bound of $n$ (c.f. Lemma \ref{lemma events}). So let us assume that such meetings do not occur. Therefore in each meeting after $T_3$, $r_3$ with state \texttt{return} catches $r_2$ with state \texttt{forward}. If at the time of any meeting after $T_3$, $r_3$ is aware of $n$ (by meeting $r_1$ twice in a run with $state$ \texttt{bounce}), then we are done. This is because $r_2$ finds out $n$ from $r_3$ at the time of the meeting and $r_1$ also had already found out $n$ from $r_3$. So assume that this does not happen at any of the meetings. Notice that if we can show that one such meeting takes place after $r_2$ visits $r_1$, then we are done. To see this observe that if $r_2$ visits $r_1$ then $r_2$ finds out $n$ and also informs $r_1$ about it. Then $r_3$ also comes to know about $n$ in the next meeting between $r_2$ and $r_3$. So we complete the proof by showing that a meeting between $r_2$ and $r_3$ takes place within $16n$ rounds from $T_3$ with $r_2$ having already met $r_1$. 

Recall that $t_i$ denotes the time between the $i$th and $(i+1)$th meeting. Let $fs_{i}$, $bs_{i}$ and $rs_{i}$ denote the number of successful steps made during this time by $r_2$ with state \texttt{forward}, $r_3$ with state \texttt{bounce} and $r_3$ with state \texttt{return} respectively. First we claim that if $r_2$ and $r_3$ swap over an edge between their $i$th and $(i+1)$th meeting, then $r_2$ will find out $n$ before the $(i+1)$th meeting. Since the $(i+1)$th meeting is not of the type Event $3$, we have $rs_i > 2bs_i$. If $bs_i > n$, then $rs_i > 2n > n$. This means that $r_3$ with state \texttt{return} meet $r_1$. But then $r_2$ must have met $r_1$ before this meeting and found out $n$. So let $bs_i \leq n$. Then it is not difficult to see that $bs_i + fs_i = n + rs_i$. Again using $rs_i > 2bs_i$, we have $fs_i > n + bs_i > n$. This implies that $r_2$ have met $r_1$ and found out $n$. So assume that $r_2$ and $r_3$ do not swap in between the $i$th and $(i+1)$th meeting. Then we have  $rs_{i} = bs_{i} + fs_{i}$. Furthermore we have $t_i \leq 2bs_{i} + 1 + fs_{i} + rs_{i}$. Here $2bs_{i} + 1$ is an upper bound on the time needed by $r_3$ to switch state from \texttt{bounce} to \texttt{return}. To see this, recall that $r_3$ changes its state to \texttt{return} if  one of the following takes place: (1) it finds out $n$ by meeting $r_1$ twice, (2) $r_3.BBlocked > r_3.BSteps$ is satisfied. Since we assumed that (1) does not take place, $r_3$ must have changed its state to \texttt{return} because of (2). Since this happens when the number of failed attempts to move exceeds the number of successful moves, we can conclude that $2bs_{i} + 1$ is an upper bound on the time needed by $r_3$ to switch state from \texttt{bounce} to \texttt{return}. Also, $fs_{i} + rs_{i}$ is an upper bound on the time needed for $r_3$ with state \texttt{return} to catch $r_2$. This is because the number of successful moves by $r_3$ is $rs_{i}$ and the number of failed attempts to move by $r_3$ is bounded by $fs_{i}$ since each failed move by $r_3$ implies a successful move by $r_2$. Substituting the value of $rs_{i}$ in the inequality we get, $t_i \leq 3bs_{i} + 2fs_{i} +1$. Since this meeting is not of the type Event $3$, we have $rs_i > 2bs_i$ $\implies fs_i > bs_i$. Therefore $t_i < 6fs_{i}$. Let $k$ be the smallest integer such that $t_1 + \ldots + t_k \geq 6n$. So we have $6n \leq$ $t_1 + t_2 + \ldots + t_k < 6(fs_{1} + fs_{2} + \ldots fs_{k})$ $\implies$ $n < fs_{1} + fs_{2} + \ldots fs_{k}$. This implies that $r_2$ makes enough progress to meet $r_1$ before the $(k+1)$th meeting. Therefore after the $(k+1)$th meeting all three robots have found out $n$. It remains to show that this meeting takes place within $16n$ rounds from $T_3$. For this observe the following. It follows from the definition of $k$ (which implies that $t_1 + \ldots + t_{k-1 } < 6n$) and the fact that $t_k \leq 10n$ that $t_1 + \ldots + t_k < 16n$.

% Therefore at the $(k+1)$th meeting (which takes place within $13n$ rounds from $T_2$ ) $r_3$ comes to know $n$ from $r_2$. Also $r_1$ has already found out $n$ from $r_2$. This completes the proof that within $13n$ rounds from $T_3$ either all three agents find out $n$ or $r_2, r_3$ agree on an upper bound of $n$.  

\end{proof}

\begin{lemma}\label{nineteen}
Within $19n$ rounds after $T_3$ exploration of the ring is complete and all three agents terminate.
\end{lemma}

\begin{proof}
By Lemma \ref{sixteen} within $16n$ rounds after $T_3$ either all three agents find out $n$ or $r_2$, $r_3$ agree on an upper bound of $n$. In the former case the agents will terminate when $TTime = 16n$, i.e., after $16n$ rounds from $T_3$. Now in the latter case when $r_2$ and $r_3$  meet at a node to agree on an upper bound $N$ of $n$, they will move in opposite direction for $N$ more rounds. Within these $N$ rounds one of them will meet $r_1$ and both will terminate. The other agents will terminate after the completion of these $N$ rounds. Recall that by Lemma \ref{lemma events} $N<3n$. Hence in this case all the three agents terminate within $19n$ rounds after $T_3$. 

% In Lemma \ref{sixteen} it is shown that a meeting takes place between $r_2$ and $r_3$ within $16n$ rounds from $T_3$ with $r_2$ having already met $r_1$. When $r_2$ meets $r_1$ both of them comes to know about $n$. Hence when $r_2$ meets $r_3$ after that both of them terminate. The landmark $r_1$ also knows about $n$ and round $T_3$ through $TTime$. Hence $r_1$ also terminates $19n$ rounds after $T_3$.

\end{proof}

\begin{theorem}
\textsc{Exploration} with explicit termination is solvable by three agents with chirality in $2^{k+\lceil\log n\rceil + 3} + 23n = O(2^kn)$ rounds.
\end{theorem}

\begin{proof}
By Lemma \ref{nineteen} the exploration is complete and all three agents terminate within $T_3 + 19n$ rounds. By Lemma \ref{lemma four} $T_3 \leq T_2 + 4n$ and by Lemma \ref{lemma T2} $T_2 < 2^{k+\lceil\log n\rceil + 3}$. Therefore our algorithm solves  \textsc{Exploration} with explicit termination by three agents with chirality in $2^{k+\lceil\log n\rceil + 3} + 23n = O(2^kn)$ rounds.

\end{proof}

\section{Exploration by Agents without Chirality}\label{Secwoch}

\subsection{Contiguous Agreement}\label{ConAg}
In this section we define a new problem called \textsc{Contiguous Agreement}. Three agents with unique identifiers are placed at three different nodes in the ring. In each round, each agent chooses a direction according to a deterministic algorithm based on its ID  and current round. The requirement of the problem is that the agents have to choose the same direction for some $N$ consecutive rounds where $N$ is a constant unknown to the agents.

\subsubsection{The Algorithm}
Before presenting the algorithm, we describe the construction of modified identifiers which will be used in the algorithm. Recall that $r.ID$ is a binary string of length $k$. We now describe the construction of the corresponding modified identifier $r.MID$ which is a binary string of length $\frac{k(k-1)}{2} + k + 1$. We shall first concatenate a string of length $\frac{k(k-1)}{2}$ at end of $r.ID$. Let us write $\frac{k(k-1)}{2} = l$. To define the string, we shall identify each position of the string as, instead of an integer from $[l] = \{ 1, \ldots,l \}$, a 2-tuple from the set $S = \{(u,v) \in [k]\times [k] \mid u<v\}$. In order to formally describe this, let us  define a bijection $\phi : S \rightarrow [l]$ in the following way. Notice that $|S|=l$. Arrange the elements of $S$ in lexicographic order. For any $(u,v) \in S$, we define $\phi((u,v))$ to be the position of $(u,v)$ in this arrangement. For example, if $k=4$, then the elements of $S$, arranged in lexicographic order, are (1,2), (1,3), (1,4), (2,3), (2,4), (3,4). Therefore, we have $\phi((1,2)) = 1$, $\phi((1,3)) = 2$, $\phi((1,4)) = 3$, $\phi((2,3)) = 4$, $\phi((2,4)) = 5$ and $\phi((3,4)) = 6$. Now we define the string of length $l$ that will be concatenated with  $r.ID$. The $i$th bit of the string is the $\mathbb{Z}_2$ sum of the $u$th and $v$th bit of $r.ID$ where $(u,v) = \phi^{-1}(i)$. After the concatenation, we get a string of length $k+l = k + \frac{k(k-1)}{2}$. Finally we append  0 at the beginning of this  string to obtain $r.MID$ of length $\frac{k(k-1)}{2} + k + 1$ (c.f. Fig. \ref{fig: mod}).

\begin{figure}[htb!]
\centering
\fontsize{8pt}{8pt}\selectfont
\def\svgwidth{.9\textwidth}
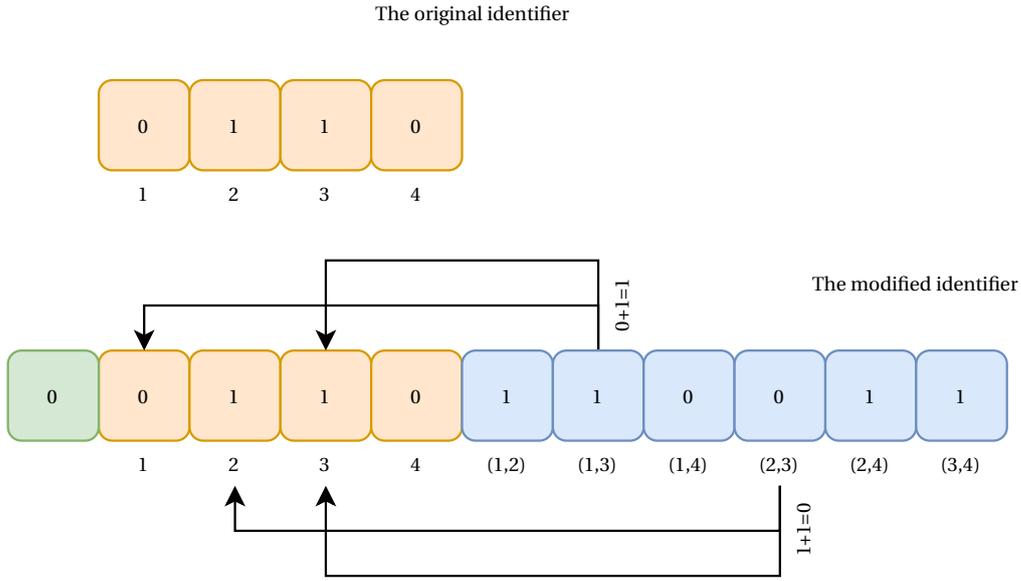
\caption[]{The construction of the modified identifier.}
\label{fig: mod}
\end{figure}

We now present the algorithm that solves \textsc{Contiguous Agreement}. The algorithm works in phases with the length of the phases being $2^{j}\left(\frac{k(k-1)}{2} + k + 1 \right)$, $j= 0,1,2, \ldots$. For a string $S$ and a positive integer $t$, let $Dup(S, t)$ denote the string obtained by repeating  each bit of string $S$  $t$ times. For example, $Dup(101, 3) = 111000111$. For the $j$th phase, the agent $r$ computes $Dup(r.MID, 2^{j})$. Notice that the length of the $j$th phase is equal to the length of $Dup(r.MID, 2^{j})$.
In the $i$th round of the $j$th phase, $r$ moves left if the $i$th bit of $Dup(r.MID, 2^{j})$ is $0$ and otherwise moves right.

\subsubsection{Correctness}

\begin{theorem}\label{thm : ca}
The algorithm described above solves \textsc{Contiguous Agreement}.
\end{theorem}

\begin{proof}
Let us denote the three agents by $r_1, r_2$ and $r_3$. For a binary string $S$ and an index $1 \leq \alpha \leq |S|$, we denote by $S[\alpha]$ the $\alpha$th bit of $S$.             We first show that in the $0$th phase, there is a round in which all three agents choose the same direction. In the $0$th phase, each agent $r_i$ uses the string $Dup(r_i.MID, 2^{0}) = r_i.MID$ to set its direction. First consider the case where all three agree on the orientation, say the left according to each of the agents is the counterclockwise direction. Clearly if there is an index $\alpha$ where the strings $r_1.MID$, $r_2.MID$ and $r_3.MID$ have the same bit, i.e., $r_1.MID[\alpha]$ $= r_2.MID[\alpha]$ $= r_3.MID[\alpha]$, then the agents will decide the same the direction in the $\alpha$th round of the $0$th phase. Now suppose that the agents do not agree on orientation. Consider the case where  left according to $r_1, r_2$ is the counterclockwise direction (i.e., $r_1$ and $r_2$ have agreement on  orientation), while left according to $r_3$ is the clockwise direction. Clearly in this case, if there is an index $\alpha$ such that $r_1.MID[\alpha]$ $= r_2.MID[\alpha]$ $\neq r_3.MID[\alpha]$, then the agents will decide the same the direction in the $\alpha$th round of the $0$th phase. Therefore, it follows from the above discussion that it suffices to show that $\exists$ indices $\alpha$, $\beta$, $\gamma$ and $\eta$ such that,
\begin{enumerate}
    \item $r_1.MID[\alpha]$ $= r_2.MID[\alpha]$ $= r_3.MID[\alpha]$
    \item $r_1.MID[\beta]$ $= r_2.MID[\beta]$ $\neq r_3.MID[\beta]$
    \item $r_1.MID[\gamma]$ $= r_3.MID[\gamma]$ $\neq r_2.MID[\gamma]$
    \item $r_2.MID[\eta]$ $= r_3.MID[\eta]$ $\neq r_1.MID[\eta]$
\end{enumerate}

% Since $r_1.ID$, $r_2.ID$ and $r_3.ID$ are all distinct, $r_1.MID$, $r_2.MID$ and $r_3.MID$ are also distinct. 

% For brevity, let us write $I_1 = r_1.MID$, $I_2 =r_2.MID$ and $I_3 =r_3.MID$. 

Recall that each of the strings start with $0$. Hence, we have $\alpha=1$. Since  $r_1.ID \neq r_2.ID$, $\exists$ an index $a$ so that $r_1.ID[a] \neq r_2.ID[a]$. Without loss of generality let  $r_3.ID[a] = r_1.ID[a]$, i.e., $r_1.ID[a] = r_3.ID[a] \neq r_2.ID[a]$. So we let $\gamma = a+1$ which fulfills the requirement that $r_1.MID[\gamma]$ $= r_3.MID[\gamma]$ $\neq r_2.MID[\gamma]$. We add 1 here because of the  0 appended at the beginning of the MIDs. Similarly since $r_1.ID \neq r_3.ID$, $\exists$ an index $b \neq a$ so that $r_1.ID[b] \neq r_3.ID[b]$. Without loss of generality let  $r_2.ID[b] = r_1.ID[b]$, i.e., $r_1.ID[b] = r_2.ID[b] \neq r_3.ID[b]$. So we let $\beta = b+1$. So now it remains to find the index $\eta$. We take $\eta$ to be the index of MID where we put the sum of the $\beta$th bit of MID ($\beta-1$th bit of ID) and the $\gamma$th bit of MID ($\gamma-1$th bit of ID). We have to show that $r_2.MID[\eta]$ $= r_3.MID[\eta]$ $\neq r_1.MID[\eta]$. For the equality, observe the following. 

\begin{align*}
r_2.MID[\eta] &= r_2.MID[\beta] +  r_2.MID[\gamma]\\
% &= \overline{r_3.MID[\beta]} +  \overline{r_3.MID[\gamma]} && \text{for } x,y \in \mathbb{Z}_2, x \neq y \implies x = \overline{y} = y+1 \\
&= (1+ r_3.MID[\beta]) + (1+ r_3.MID[\gamma]) && (\text{for } x,y \in \mathbb{Z}_2, x \neq y \iff x =  y+1) \\
&= r_3.MID[\beta] + r_3.MID[\gamma]\\
&= r_3.MID[\eta]\\
\end{align*}

To prove the inequality, we assume for the sake of contradiction that $r_1.MID[\eta] = r_2.MID[\eta]$. This leads to a contradiction as shown in the following.

\begin{align*}
r_1.MID[\eta] &= r_2.MID[\eta]\\
\implies r_1.MID[\beta] +  r_1.MID[\gamma] &= r_2.MID[\beta] +  r_2.MID[\gamma]\\
\implies r_2.MID[\beta] +  (1 + r_2.MID[\gamma]) &= r_2.MID[\beta] +  r_2.MID[\gamma]\\
\implies 1 &= 0\\
\end{align*}

% Also for similar reason $\exists$ index $j \neq i$ where $A_{j} \neq C_{j}$.
% Now two cases arise:
% \begin{enumerate}
%     \item Case-I: $A_{j} = B_{j} \neq C_{j}$
%     \item Case-II: $C_{j} = B_{j} \neq A_{j}}$
% \end{enumerate}

% Let $k$ be the index where we put the sum of the $i$th and the $j$th index.

% Case-I:

% $i$th position $\rightarrow$ $A_{i} = C_{i} \neq B_{i}$

% $j$th position $\rightarrow$ $A_{j} = B_{j} \neq C_{j}$

% $k$th position $\rightarrow$ $A_{k} = A_{i} + A_{j}$,
                             
% Now, $B_{k} = B_{i} + B_{j} = \overline{C_{i}} + \overline{C_{j}} = 1+ C_{j} + 1 + C_{j} = C_{i} + C_{j} = C_{k}$.

% We claim that $A_{k} \neq B_{k}$. For the sake of contradiction let $A_{k} \neq B_{k}$.

% $\therefore$ $A_{k} = B_{k}$

% $\implies$ $A_{i} + A_{j} = B_{i} + B_{j}$

% $\implies$ $\overline{B_{i}} + B_{j} = B_{i} + B_{j}$

% $\implies$ $\overline{B_{i}} = B_{i}$

% $\therefore$ $B_{k} = C_{k} \neq A_{k}$

% Similarly if we consider Case-II, we shall get, $A_{k} = B_{k} \neq C_{k}$.

Hence, we show that there is round in the $0$th phase where all three agents will decide the same  direction. It immediately follows from the proof that there are $2^j$ consecutive rounds in the $j$th phase where all three agents will choose the same  direction. Hence, for $j = \lceil\log N\rceil$, the agents will choose the same direction for $N$ consecutive rounds in the $j$th phase.
\end{proof}

\subsection{Meeting by Agents without Chirality}\label{achiralmeet}

In Section \ref{MeetChiral}, we describe an algorithm that solves \textsc{Meeting} by agents with chirality. In the current setting where agents do not have chirality, the main idea is to use the strategy of \textsc{Contiguous Agreement} so that the agents can implicitly agree on a common direction and  solve \textsc{Meeting} by employing the strategy from section \ref{MeetChiral}. Similar to the algorithm for \textsc{Contiguous Agreement}, our algorithm for \textsc{Meeting} also works in phases. In the algorithm for \textsc{Contiguous Agreement} the length of the phases were $2^{j}\left(\frac{k(k-1)}{2} + k + 1 \right)$, $j= 0,1,2, \ldots$. For \textsc{Meeting}, the phases will be of length $2^{j+k}\left(\frac{k(k-1)}{2} + k + 1 \right)$, $j= 0,1,2, \ldots$. In the $j$th phase of the algorithm an agent $r$ uses the string $Dup(r.MID, 2^{j+k})$ to decide its movement. Notice that the length of  $Dup(r.MID, 2^{j+k})$ is equal to the length of the $j$th phase. The string $Dup(r.MID, 2^{j+k})$ is a concatenation of $\left(\frac{k(k-1)}{2} + k + 1 \right)$ blocks of length $2^{j+k}$ where each block consists of all 0's or all 1's. Our plan is to simulate the \textsc{Meeting} algorithm from section \ref{MeetChiral}. So, in the $2^{j+k}$ rounds 
corresponding to each block, the agent $r$ will (try to) move in the first $val(r.ID)2^{j}$ rounds and will be stationary for the remaining $\left(2^{k} - val(r.ID)\right)2^{j}$ rounds. If the block consists of $0$'s, then the movement will be towards left and otherwise towards right.

\begin{theorem}\label{ th meetwo}
The above algorithm solves \textsc{Meeting} for  three agents without chirality. The algorithm ensures that the meeting takes place within $k^2 2^{k+\lceil\log n\rceil + 3}$ rounds.
\end{theorem}

\begin{proof}
Let $p = \lceil\log 2n\rceil $. If two agents have met before the $p$th phase, then we are done. Otherwise,  we show that a pair of agents must meet during the $p$th phase. Recall that in the $p$th phase, an agent $r$ uses the bits of the string $Dup(r.MID, 2^{p+k})$ to decide its movement in each round. Each block of $Dup(r.MID, 2^{p+k})$ is $2^{p+k}$ bits long. From the proof of Theorem \ref{thm : ca} it follows that there is a block in $Dup(r.MID, 2^{p+k})$ so that the agents have an agreement in direction in the rounds corresponding to that block. Recall that the agents simulate the \textsc{Meeting} algorithm from section \ref{MeetChiral} in the rounds corresponding to each block. Hence it follows from Theorem \ref{thm: meet chiral} that two agents are guaranteed to meet during the rounds corresponding to the aforesaid block. Therefore, the algorithm ensures that the meeting takes place within $2^k\left(\frac{k(k-1)}{2} + k + 1 \right)\sum_{i=0}^{p} 2^i < k^2 2^{k+\lceil\log n\rceil + 3}$ rounds.
\end{proof}

% In the $j$th phase of the algorithm for Contiguous Agreement, an agent $r$ only decided directions based on the bits of $Dup(r.MID, 2^j)$. In case of the algorithm for meeting

\subsection{Exploration with Termination by Agents without Chirality}

For simplicity assume that the agents are placed arbitrarily at distinct nodes of the ring. We shall remove this assumption at the end of this section.
Initially the $state$ variable of all three agents are set to \texttt{search}. We shall adopt a strategy similar to the one used in Section \ref{Exprwc}. In fact, we only need to make some modifications in the algorithms to be executed by the agents with $state$ \texttt{search} and \texttt{settled}.

     The agents will execute the \textsc{Meeting} algorithm described in Section \ref{achiralmeet} until another agent is encountered. The agents will keep count of the number of rounds since the beginning in the variable $STime$. Now by Theorem \ref{ th meetwo} two of the agents are guaranteed to meet within $k^2 2^{k+\lceil\log n\rceil + 3}$ rounds. Upon meeting the agents will agree on a common direction, say the right direction of the agent with larger ID. Without loss of generality assume that the agreed direction is the clockwise direction.      The agent with smaller ID, say $r_1$, will become the \texttt{settled} agent and the one with larger ID, say $r_2$, will continue moving in the  clockwise direction. The agent $r_1$ will save the port number leading to the agreed direction, i.e., clockwise.   Let $r_3$ denote the third agent which is still executing the \textsc{Meeting} algorithm. Using similar arguments as in Lemma \ref{lemma T2} we can show that a second meeting is guaranteed to take place on or before $ \lceil\log 2n\rceil + 1$th phase, i.e., within $k^2 2^{k+\lceil\log n\rceil + 4}$ rounds from the start of the algorithm. However unlike in Section \ref{Exprwc} where the agents had chirality, here the second meeting may also take place between $r_1$ and $r_2$. This is because $r_3$ is moving in clockwise direction in some rounds and counterclockwise in other rounds. Hence there is a possibility that $r_2$ and $r_3$ may swap over an edge and $r_1$ meets $r_2$ first. However even then it is not difficult to see that $r_3$ is guaranteed to meet $r_1$ or $r_2$ on or before ($ \lceil\log 2n\rceil+1)$th phase i.e., within $k^2 2^{k+\lceil\log n\rceil + 4}$ rounds from the start of the algorithm. To see this, observe that $r_2$ and $r_3$ will try to move in the same direction for some $4n$ consecutive rounds in the ($ \lceil\log 2n\rceil+1) = \lceil\log 4n\rceil$th phase. Within the first $2n$ rounds a meeting should take place by Lemma \ref{lemma main}. If this meeting involves $r_3$ then we are done. Otherwise $r_1$ and $r_2$ meet each other and then by again applying by Lemma \ref{lemma main}, $r_3$ will meet one of them within the following $2n$ rounds.

     Now consider the following cases depending on which of the two robots meet on the second meeting.

\begin{enumerate}
    \item Suppose that the second meeting takes place between $r_1$ and $r_2$. In this case the ring has been explored and $r_2$ finds out $n$ and informs $r_1$ about it. Then $r_2$ will continue moving in the clockwise direction. Both agents will terminate after the round when $STime = k^2 2^{k+\lceil\log n\rceil + 4}$. Recall that $r_3$ is guaranteed to meet one of them in the meantime and will terminate immediately.
    
    \item Suppose that the second meeting takes place between $r_2$ and $r_3$. Then $r_2$ informs $r_3$ about the agreed direction. Hence the case reduces to the setting of Section \ref{Exprwc}. Therefore $r_2$ and $r_3$ will change their $state$ to \texttt{forward} and \texttt{bounce} respectively and execute the algorithms as before.
    
    \item Now suppose that the second meeting takes place between $r_1$ and $r_3$. In this case $r_3$ will come to know about the agreed direction and again the case reduces to the setting of Section \ref{Exprwc}. So $r_3$ will continue to move in the agreed direction i.e., clockwise.
\end{enumerate}

\begin{theorem}
\textsc{Exploration} with explicit termination is solvable by three agents without chirality in $  k^2 2^{k+\lceil\log n\rceil + 4}  + 23n = O(k^2 2^k n)$ rounds.
\end{theorem}

\begin{proof}
 It follows from the above discussions and the results proved in Section \ref{Exprwc} that  the agents will terminate after exploring the ring within  $  k^2 2^{k+\lceil\log n\rceil + 4}  + 23n = O(k^2 2^kn)$ rounds.

\end{proof}

Recall that we assumed that  the agents are placed initially at distinct nodes of the ring. We now show that this assumption is not necessary if the initial configuration has two agents $r_1$, $r_2$ at the same node and the third agent $r_3$ at a different node. Then the case reduces to the situation when the first meeting takes place. Then $r_1$ and $r_2$ will change their $state$ to \texttt{settled} and \texttt{forward} while $r_3$ will execute the \textsc{Meeting} algorithm with $state$ \texttt{search}.
The algorithm will progress as before and achieve exploration with termination.

If all three agents are in the same node in the initial configuration then the three agents will compare their identifiers and will change their $state$ to \texttt{settled}, \texttt{forward} and \texttt{bounce} accordingly. Again, the algorithm will progress as before and achieve exploration with termination.

\section{Concluding Remarks}

We showed that \textsc{Exploration} (with explicit termination) in a dynamic always connected ring without any landmark node is solvable by three non-anonymous agents without chirality. This is optimal in terms of the number of agents used as the problem is known to be unsolvable by two agents. Our algorithm takes $O(k^2 2^k n)$ rounds to solve the problem where $n$ is the size of the ring and $k$ is the length of the identifiers of the agents. An interesting question is whether the problem can be solved in $O(poly(k)n)$ rounds. However with $k = O(1)$ the round complexity is $O(n)$. This is asymptotically optimal as there are $n$ nodes to be explored and in each round three agents can visit at most three nodes.

A challenging problem that remains open is \textsc{Exploration} in a dynamic network of arbitrary underlying topology. Except for some bounds on the number of agents  obtained in the recent work \cite{GOTOH20211}, almost nothing is known. As for other special topologies only the case of torus has been studied \cite{DBLP:conf/icdcs/GotohSOKM18}. Comparison of what can be achieved by anonymous and non-anonymous agents in these settings is an direction to explore.

\bibliographystyle{plainurl}
\bibliography{dynamic}

\begin{thebibliography}{10}

\bibitem{AgarwallaAMKS18}
Ankush Agarwalla, John Augustine, William K.~Moses Jr., Sankar~Madhav K., and
  Arvind~Krishna Sridhar.
\newblock Deterministic dispersion of mobile robots in dynamic rings.
\newblock In {\em Proceedings of the 19th International Conference on
  Distributed Computing and Networking, {ICDCN} 2018, Varanasi, India, January
  4-7, 2018}, pages 19:1--19:4. {ACM}, 2018.
\newblock \href {https://doi.org/10.1145/3154273.3154294}
  {\path{doi:10.1145/3154273.3154294}}.

\bibitem{AlbersH00}
Susanne Albers and Monika~Rauch Henzinger.
\newblock Exploring unknown environments.
\newblock {\em {SIAM} J. Comput.}, 29(4):1164--1188, 2000.
\newblock \href {https://doi.org/10.1137/S009753979732428X}
  {\path{doi:10.1137/S009753979732428X}}.

\bibitem{BournatDP18}
Marjorie Bournat, Swan Dubois, and Franck Petit.
\newblock Gracefully degrading gathering in dynamic rings.
\newblock In {\em 20th International Symposium on Stabilization, Safety, and
  Security of Distributed Systems, {SSS} 2018, Tokyo, Japan, November 4-7,
  2018, Proceedings}, volume 11201 of {\em Lecture Notes in Computer Science},
  pages 349--364. Springer, 2018.
\newblock \href {https://doi.org/10.1007/978-3-030-03232-6\_23}
  {\path{doi:10.1007/978-3-030-03232-6\_23}}.

\bibitem{CasteigtsFQS12}
Arnaud Casteigts, Paola Flocchini, Walter Quattrociocchi, and Nicola Santoro.
\newblock Time-varying graphs and dynamic networks.
\newblock {\em Int. J. Parallel Emergent Distributed Syst.}, 27(5):387--408,
  2012.
\newblock \href {https://doi.org/10.1080/17445760.2012.668546}
  {\path{doi:10.1080/17445760.2012.668546}}.

\bibitem{DasFKNS07}
Shantanu Das, Paola Flocchini, Shay Kutten, Amiya Nayak, and Nicola Santoro.
\newblock Map construction of unknown graphs by multiple agents.
\newblock {\em Theor. Comput. Sci.}, 385(1-3):34--48, 2007.
\newblock \href {https://doi.org/10.1016/j.tcs.2007.05.011}
  {\path{doi:10.1016/j.tcs.2007.05.011}}.

\bibitem{DBLP:conf/sofsem/DasLG19}
Shantanu Das, Giuseppe Antonio~Di Luna, and Leszek~Antoni Gasieniec.
\newblock Patrolling on dynamic ring networks.
\newblock In {\em 45th International Conference on Current Trends in Theory and
  Practice of Computer Science, {SOFSEM} 2019, Nov{\'{y}} Smokovec, Slovakia,
  January 27-30, 2019, Proceedings}, volume 11376 of {\em Lecture Notes in
  Computer Science}, pages 150--163. Springer, 2019.
\newblock \href {https://doi.org/10.1007/978-3-030-10801-4\_13}
  {\path{doi:10.1007/978-3-030-10801-4\_13}}.

\bibitem{DengP99}
Xiaotie Deng and Christos~H. Papadimitriou.
\newblock Exploring an unknown graph.
\newblock {\em J. Graph Theory}, 32(3):265--297, 1999.

\bibitem{DieudonneP14}
Yoann Dieudonn{\'{e}} and Andrzej Pelc.
\newblock Deterministic network exploration by anonymous silent agents with
  local traffic reports.
\newblock {\em {ACM} Trans. Algorithms}, 11(2):10:1--10:29, 2014.
\newblock \href {https://doi.org/10.1145/2594581} {\path{doi:10.1145/2594581}}.

\bibitem{Erlebach0K15}
Thomas Erlebach, Michael Hoffmann, and Frank Kammer.
\newblock On temporal graph exploration.
\newblock In {\em 42nd International Colloquium on Automata, Languages, and
  Programming, {ICALP} 2015, Kyoto, Japan, July 6-10, 2015, Proceedings, Part
  {I}}, volume 9134 of {\em Lecture Notes in Computer Science}, pages 444--455.
  Springer, 2015.
\newblock \href {https://doi.org/10.1007/978-3-662-47672-7\_36}
  {\path{doi:10.1007/978-3-662-47672-7\_36}}.

\bibitem{FlocchiniKMS09}
Paola Flocchini, Matthew Kellett, Peter~C. Mason, and Nicola Santoro.
\newblock Map construction and exploration by mobile agents scattered in a
  dangerous network.
\newblock In {\em 23rd {IEEE} International Symposium on Parallel and
  Distributed Processing, {IPDPS} 2009, Rome, Italy, May 23-29, 2009}, pages
  1--10. {IEEE}, 2009.
\newblock \href {https://doi.org/10.1109/IPDPS.2009.5161080}
  {\path{doi:10.1109/IPDPS.2009.5161080}}.

\bibitem{FlocchiniMS13}
Paola Flocchini, Bernard Mans, and Nicola Santoro.
\newblock On the exploration of time-varying networks.
\newblock {\em Theor. Comput. Sci.}, 469:53--68, 2013.
\newblock \href {https://doi.org/10.1016/j.tcs.2012.10.029}
  {\path{doi:10.1016/j.tcs.2012.10.029}}.

\bibitem{FraigniaudIPPP05}
Pierre Fraigniaud, David Ilcinkas, Guy Peer, Andrzej Pelc, and David Peleg.
\newblock Graph exploration by a finite automaton.
\newblock {\em Theor. Comput. Sci.}, 345(2-3):331--344, 2005.
\newblock \href {https://doi.org/10.1016/j.tcs.2005.07.014}
  {\path{doi:10.1016/j.tcs.2005.07.014}}.

\bibitem{GOTOH20211}
Tsuyoshi Gotoh, Paola Flocchini, Toshimitsu Masuzawa, and Nicola Santoro.
\newblock Exploration of dynamic networks: Tight bounds on the number of
  agents.
\newblock {\em Journal of Computer and System Sciences}, 122:1--18, 2021.
\newblock \href {https://doi.org/10.1016/j.jcss.2021.04.003}
  {\path{doi:10.1016/j.jcss.2021.04.003}}.

\bibitem{DBLP:conf/icdcs/GotohSOKM18}
Tsuyoshi Gotoh, Yuichi Sudo, Fukuhito Ooshita, Hirotsugu Kakugawa, and
  Toshimitsu Masuzawa.
\newblock Group exploration of dynamic tori.
\newblock In {\em 38th {IEEE} International Conference on Distributed Computing
  Systems, {ICDCS} 2018, Vienna, Austria, July 2-6, 2018}, pages 775--785.
  {IEEE} Computer Society, 2018.
\newblock \href {https://doi.org/10.1109/ICDCS.2018.00080}
  {\path{doi:10.1109/ICDCS.2018.00080}}.

\bibitem{IlcinkasKW14}
David Ilcinkas, Ralf Klasing, and Ahmed~Mouhamadou Wade.
\newblock Exploration of constantly connected dynamic graphs based on cactuses.
\newblock In {\em 21st International Colloquium on Structural Information and
  Communication Complexity, {SIROCCO} 2014, Takayama, Japan, July 23-25, 2014.
  Proceedings}, volume 8576 of {\em Lecture Notes in Computer Science}, pages
  250--262. Springer, 2014.
\newblock \href {https://doi.org/10.1007/978-3-319-09620-9\_20}
  {\path{doi:10.1007/978-3-319-09620-9\_20}}.

\bibitem{IlcinkasW11}
David Ilcinkas and Ahmed~Mouhamadou Wade.
\newblock On the power of waiting when exploring public transportation systems.
\newblock In {\em 15th International Conference on Principles of Distributed
  Systems, {OPODIS} 2011, Toulouse, France, December 13-16, 2011. Proceedings},
  volume 7109 of {\em Lecture Notes in Computer Science}, pages 451--464.
  Springer, 2011.
\newblock \href {https://doi.org/10.1007/978-3-642-25873-2\_31}
  {\path{doi:10.1007/978-3-642-25873-2\_31}}.

\bibitem{IlcinkasW13}
David Ilcinkas and Ahmed~Mouhamadou Wade.
\newblock Exploration of the {T}-interval-connected dynamic graphs: The case of
  the ring.
\newblock In {\em 20th International Colloquium on Structural Information and
  Communication Complexity, {SIROCCO} 2013, Ischia, Italy, July 1-3, 2013},
  volume 8179 of {\em Lecture Notes in Computer Science}, pages 13--23.
  Springer, 2013.
\newblock \href {https://doi.org/10.1007/978-3-319-03578-9\_2}
  {\path{doi:10.1007/978-3-319-03578-9\_2}}.

\bibitem{KshemkalyaniMS20}
Ajay~D. Kshemkalyani, Anisur~Rahaman Molla, and Gokarna Sharma.
\newblock Efficient dispersion of mobile robots on dynamic graphs.
\newblock In {\em 40th {IEEE} International Conference on Distributed Computing
  Systems, {ICDCS} 2020, Singapore, November 29 - December 1, 2020}, pages
  732--742. {IEEE}, 2020.
\newblock \href {https://doi.org/10.1109/ICDCS47774.2020.00100}
  {\path{doi:10.1109/ICDCS47774.2020.00100}}.

\bibitem{DBLP:journals/sigact/KuhnO11}
Fabian Kuhn and Rotem Oshman.
\newblock Dynamic networks: models and algorithms.
\newblock {\em {SIGACT} News}, 42(1):82--96, 2011.
\newblock \href {https://doi.org/10.1145/1959045.1959064}
  {\path{doi:10.1145/1959045.1959064}}.

\bibitem{dc/LunaDFS20}
Giuseppe Antonio~Di Luna, Stefan Dobrev, Paola Flocchini, and Nicola Santoro.
\newblock Distributed exploration of dynamic rings.
\newblock {\em Distributed Comput.}, 33(1):41--67, 2020.
\newblock \href {https://doi.org/10.1007/s00446-018-0339-1}
  {\path{doi:10.1007/s00446-018-0339-1}}.

\bibitem{LunaFPPSV20}
Giuseppe Antonio~Di Luna, Paola Flocchini, Linda Pagli, Giuseppe Prencipe,
  Nicola Santoro, and Giovanni Viglietta.
\newblock Gathering in dynamic rings.
\newblock {\em Theor. Comput. Sci.}, 811:79--98, 2020.
\newblock \href {https://doi.org/10.1016/j.tcs.2018.10.018}
  {\path{doi:10.1016/j.tcs.2018.10.018}}.

\bibitem{MandalMM20}
Subhrangsu Mandal, Anisur~Rahaman Molla, and William K.~Moses Jr.
\newblock Live exploration with mobile robots in a dynamic ring, revisited.
\newblock In {\em 16th International Symposium on Algorithms and Experiments
  for Wireless Sensor Networks, {ALGOSENSORS} 2020, Pisa, Italy, September
  9-10, 2020}, volume 12503 of {\em Lecture Notes in Computer Science}, pages
  92--107. Springer, 2020.
\newblock \href {https://doi.org/10.1007/978-3-030-62401-9\_7}
  {\path{doi:10.1007/978-3-030-62401-9\_7}}.

\bibitem{MichailS16}
Othon Michail and Paul~G. Spirakis.
\newblock Traveling salesman problems in temporal graphs.
\newblock {\em Theor. Comput. Sci.}, 634:1--23, 2016.
\newblock \href {https://doi.org/10.1016/j.tcs.2016.04.006}
  {\path{doi:10.1016/j.tcs.2016.04.006}}.

\bibitem{OoshitaD18}
Fukuhito Ooshita and Ajoy~K. Datta.
\newblock Brief announcement: Feasibility of weak gathering in
  connected-over-time dynamic rings.
\newblock In {\em 20th International Symposium on Stabilization, Safety, and
  Security of Distributed Systems, {SSS} 2018, Tokyo, Japan, November 4-7,
  2018, Proceedings}, volume 11201 of {\em Lecture Notes in Computer Science},
  pages 393--397. Springer, 2018.
\newblock \href {https://doi.org/10.1007/978-3-030-03232-6\_27}
  {\path{doi:10.1007/978-3-030-03232-6\_27}}.

\bibitem{PanaiteP99}
Petrisor Panaite and Andrzej Pelc.
\newblock Exploring unknown undirected graphs.
\newblock {\em J. Algorithms}, 33(2):281--295, 1999.
\newblock \href {https://doi.org/10.1006/jagm.1999.1043}
  {\path{doi:10.1006/jagm.1999.1043}}.

\bibitem{DBLP:conf/opodis/Santoro15}
Nicola Santoro.
\newblock Time to change: On distributed computing in dynamic networks
  ({K}eynote).
\newblock In {\em 19th International Conference on Principles of Distributed
  Systems, {OPODIS} 2015, December 14-17, 2015, Rennes, France}, volume~46 of
  {\em LIPIcs}, pages 3:1--3:14. Schloss Dagstuhl - Leibniz-Zentrum f{\"{u}}r
  Informatik, 2015.
\newblock \href {https://doi.org/10.4230/LIPIcs.OPODIS.2015.3}
  {\path{doi:10.4230/LIPIcs.OPODIS.2015.3}}.

\end{thebibliography}

\end{document}